\documentclass{elsarticle}

\usepackage{amsmath}
\usepackage{amssymb}
\usepackage{latexsym}
\usepackage{xspace}
\def\ie{i.~e., }

\newcommand{\cost}{cost}
\def\int{\ensuremath{\mathrm{INT}}\xspace}
\def\ll{\ensuremath{\log\log}}
\def\lll{\ensuremath{\log\log\log}}

\def\minTCO{{\sf Min-TCO}\xspace}
\def\dTCOk{{Min-$d$-TCO($k$)}\xspace}
\def\minHS{{\sf Min-HS}\xspace}
\def\dHS{{\sf $d$-HS}\xspace}
\def\oHS{{\sf $O(d^2)$-HS}\xspace}
\def\tHS{{\sf 2-HS}\xspace}
\def\dHSk{{$d$-HS($k$)}\xspace}
\def\minSC{{\sf Min-SC}\xspace}
\def\minVC{{\sf Min-VC}\xspace}

\newcommand{\ptas}{\ensuremath{{\cal PT\!AS}}\xspace}
\newcommand{\apx}{\ensuremath{{\cal AP\!X}}\xspace}
\newcommand{\logapx}{\ensuremath{{\cal LOG\!APX}}\xspace}
\newcommand{\npo}{\ensuremath{{\cal NPO}}\xspace}

\newcommand{\classP}{\ensuremath{{\cal P}}\xspace}
\newcommand{\classNP}{\ensuremath{{\cal NP}}\xspace}
\newcommand{\ap}{AP\xspace}

\newcommand{\ihs}{I_\mathrm{HS}}
\newcommand{\shs}{Sol_\mathrm{HS}}
\newcommand{\ohs}{Opt_\mathrm{HS}}

\newcommand{\itco}{I_\mathrm{TCO}}
\newcommand{\stco}{Sol_\mathrm{TCO}}
\newcommand{\otco}{Opt_\mathrm{TCO}}

\newcommand{\svc}{Sol_\mathrm{VC}}
\newcommand{\ovc}{Opt_\mathrm{VC}}

\newdefinition{definition}{Definition}
\newdefinition{problem}{Problem} 
\newtheorem{theorem}{Theorem}
\newtheorem{lemma}{Lemma}
\newtheorem{corollary}{Corollary}
\newproof{proof}{Proof}

\begin{document}

\begin{frontmatter}

\title{On the Approximability and Hardness of Minimum Topic Connected Overlay and Its Special 
Instances\tnoteref{thanks}\tnoteref{presented}}
\tnotetext[thanks]{This research is partly supported by the Japan Society for the Promotion of Science, 
Grant-in-Aid for Scientific Research, 21500013, 21680001, 22650004, 22700010, 23104511, 23310104, 
Foundation for the Fusion of Science Technology (FOST) and INAMORI FOUNDATION.
The research is also partially funded by SNF grant 200021-132510/1.}
\tnotetext[presented]{Some of the results of this paper were presented at MFCS 2011 and PODC 2011.}

\author[nagoya]{Jun Hosoda}
\author[eth]{Juraj Hromkovi\v c}\ead{juraj.hromkovic@inf.ethz.ch}
\author[nagoya]{Taisuke Izumi}\ead{t-izumi@nitech.ac.jp}
\author[kyushu]{Hirotaka Ono}\ead{hirotaka@en.kyushu-u.ac.jp}
\author[eth]{Monika Steinov\'a}\ead{monika.steinova@inf.ethz.ch}
\author[nagoya]{Koichi Wada}\ead{wada@nitech.ac.jp}

\address[eth]{Department of Computer Science, ETH Zurich, Switzerland}
\address[kyushu]{Department of Economic Engineering, Kyushu University, Japan}
\address[nagoya]{Graduate School of Engineering, Nagoya Institute of Technology, Japan}

\begin{abstract}
In the context of designing a scalable overlay network to support
decentralized topic-based pub/sub communication, the Minimum
Topic-Connected Overlay problem (\minTCO in short) has been investigated:
Given a set of $t$ topics and a collection of $n$ users together with the 
lists of topics they 
are interested in, the aim is to connect these users to a~network 
by a minimum number of edges such that every graph induced by users 
interested in a common topic is connected.
It is known that \minTCO is \classNP-hard and
approximable within $O(\log t)$ in polynomial time. 

In this paper, we further investigate the problem and some of its special
instances. We give various hardness results for instances where
the number of topics in which an user is interested in is bounded by 
a constant, and also for the instances where the number of users 
interested in a common topic is constant. For the latter case, we 
present a first constant approximation algorithm.
We also present some polynomial-time algorithms for very restricted 
instances of \minTCO.

\end{abstract}


\end{frontmatter}

\section{Introduction}

Recently, the spreading of social networks and other services based on sharing content
allowed the development of many-to-many communication, often supported by these services.
Publishers publish information through a logical channel that is consumed by subscribed users.
This environment is often modeled by publish/subscribe (pub/sub) systems that can be classified 
into two categories. When the channels are associated with a collection of attributes and 
the messages are delivered to a subscriber only if their attributes match user-defined 
constraints, we speak about {\em content-based} pub/sub systems.
Each channel in {\em topic-based} pub/sub systems is associated with a single
topic and the messages are distributed to the users via channels by his/her topic selection.
There are numerous implementations of pub/sub systems, for details
see~\cite{AGDSV06,BBQQP07,CRW04,CF05,RPS06,SMPD05,ZZJKK01}.

In our paper, we focus on topic-based peer-to-peer pub/sub systems.
In such a system, subscribers interested in a particular topic have to be connected without the use of
intermediate agents (such as servers). Many aspects of such a system can be studied (see~\cite{CMTV07,OR09}).
Minimizing the diameter of the overlay network can minimize the overall time in which a message is
distributed to all the subscribers. When minimizing the (average) degree of nodes in the network,
the subscribers need to maintain a smaller number of connections.
In this paper, we study the minimization of the overall number of connections in the system.
A~small number of connections may be necessary due to maintenance requirements or may be
helpful since thus information aggregated into a single message can be broadcasted to the network and thus
amortize the head count of otherwise small messages.

We study here the hardness of {\em Minimum Topic-Connected Overlay} (\minTCO) which was studied in
different scenarios in~\cite{AAR10,CMTV07,KS08,KS03}. In \minTCO, we are given a collection of 
users, a set of topics, and a user-interest assignment, we want to connect users in 
an overlay network $G$ such that all users interested in a common topic are connected and the 
overall number of edges in $G$ is minimal.
The hardness of the problem was studied in \cite{CMTV07} and \cite{AAR10}. In~\cite{CMTV07}, 
the inapproximability by a constant was proved and a logarithmic-factor approximation algorithm 
was presented. In \cite{AAR10}, the lower bound on the approximability of \minTCO was improved
to $\Omega(\log(n))$, where $n$ is the number of users. 

Moreover, we focus here on special instances of \minTCO. We study the case where, for each topic, there
is a constant number of users interested in it. 
We also consider the case where the number of topics in which any
user is interested is bounded by a constant. We believe that such restrictions on the instances 
have wide practical applications such as when a publisher has a limited number of slots for users
or the user's application limits the number of topics that he/she can follow.

In the study of the general \minTCO, we extend the method presented in~\cite{CMTV07} and design an 
approximation-preserving reduction from instances of the minimum hitting set problem to instances 
of \minTCO. This reduction does not only prove a similar lower bound as in~\cite{AAR10}, but also shows
that \minTCO is \logapx-complete and thus, concerning approximability, equivalent with such a famous 
problem as the minimum set cover. As our reduction is not blowing up the number of users interested 
in a common topic, the reduction is also an~approximation-preserving reduction for the case where
the number of users interested in a common topic is limited to a constant. 
Furthermore, we design a one-to-one reduction of these instances to special instances of the
hitting set problem. As these special instances of the hitting set problem are constantly approximable,
we immediately obtain the first approximation algorithm for our special instances.
This, together with our approximation preserving reduction, shows that
the restriction of \minTCO to such special instances is \apx-complete.
Finally, due to the one-to-one reduction and the properties of the special instances of hitting set problem,
we show the existence of a polynomial-size kernel and a non-trivial exact algorithm, all for the
instances of \minTCO where the number of users interested in a common topic is bounded by a constant.

For the case, where the number of topics of \minTCO is bounded from above 
by $(1+\varepsilon(n))^{-1}\cdot\ll n$, for
$\varepsilon(n)\ge\frac{3/2\lll n}{\ll n-3/2\lll n}$ ($n$ is the number of users), we present a 
polynomial-time algorithm that computes the optimal solution.

In the study of instances where the number of topics any user is interested in is restricted
to a constant, we show that, if this number is at most 6, \minTCO cannot be approximated within
a factor of $694/693$ in polynomial time, unless $\classP = \classNP$, even if any pair of two
users is interested in at most three common topics.

The paper is organized as follows. Section~\ref{sec:prelim} is devoted to the preliminaries
and a summary of known results. The hardness, approximation results, kernelization and
an exact algorithm for instances of \minTCO, where we limit
the number of users interested in a common topic by a constant, are discussed in 
Section~\ref{sec:hardness-topic-bounded}. This section also provides the discussion
about \logapx-completeness of the general \minTCO.
The results related to the instances of \minTCO, where
the number of topics that each user is interested in is constant, are presented in 
Section~\ref{sec:hardness-user-bounded}. Section~\ref{sec:bounded-number-of-topics} contains 
a polynomial-time algorithm that solves \minTCO when the number of topics is small.
The conclusion is provided in Section~\ref{sec:conclusion}.

\section{Preliminaries}\label{sec:prelim}
In this section, we define basic notions used throughout the paper.
We assume that the reader is familiar with notions of graph theory. 
Let $G=(V,E)$ be an~undirected graph, where $V$ is the set of vertices
and $E$ is the set of edges. Let $V(G)$ and $E(G)$ denote the set of
vertices and the set of edges of $G$, respectively. We denote by $E[S]$ the 
set of edges of $G$ in the subgraph induced by the vertices from $S \subseteq V$,
\ie $E[S] = \{\{u, v\} \in E\mid u, v \in S\}$. The graph induced by 
$S \subseteq V$ is denoted as $G[S] = (S, E[S])$. By $N[v]$ we denote
the {\em closed neighborhood} of vertex $v$, \ie 
$N[v] = \{u\in V \mid \{u,v\} \in E\} \cup \{v\}$. A graph $G$ is called 
{\em connected}, if, for any $u_1, u_{\ell} \in V$,  there exists a path 
$(u_1, u_2, \ldots, v_{\ell})$ such that $\{u_i, u_{i+1}\}\in E$, 
for all $1 \le i < \ell$. 

Let $x$ be an instance of an optimization problem (in this paper, \minTCO, %\minSC,
\minVC or \minHS), then by $|x|$ we denote the size of this instance, \ie the number 
of vertices and topics of an instance of \minTCO and the number of elements and 
sets of an instance of \minHS. For a set $S$, $|S|$ 
denotes the size of the set, \ie the number of its elements.

The set of users or nodes of our network is denoted by $U=\{u_1, u_2, \ldots, u_n\}$.
The topics are $T=\{t_1, t_2, \ldots, t_m\}$.
Each user subscribes to several topics. This relation is expressed by
the user interest function $\int:U\rightarrow2^{T}$. The set of all vertices of $U$ 
interested in a topic $t$ is denoted by $U_t$. For instance, if user $u\in U$ is 
interested in topics $t_1$, $t_3$ and $t_4$, then we have $\int(u)=\{t_1, t_3, t_4\}$ 
and $u \in U_{t_1}, U_{t_3}, U_{t_4}$. For a given set of users $U$, a set of topics $T$, 
and an interest function \int, we say that a graph $G=(U,E)$ with 
$E \subseteq \{\{u,v\} \mid u,v \in U \land u \neq v\}$ is {\em $t$-topic-connected}, 
for $t \in T$, if the subgraph $G[U_t]$ is connected. We call the graph 
{\em topic-connected} if it is $t$-topic-connected for each topic $t \in T$.
Note that the topic-connectedness property implies that a message published for topic $t$ 
is transmitted to all users interested in this topic without using non-interested users
as intermediate nodes.

The most general problem that we study in this paper is called {\em Minimum Topic Connected Overlay}:
\begin{problem} \label{def:minTCO}
\minTCO is the following optimization problem:
\begin{description}
\item[Input:] A set of users $U$, a set of topics $T$, and an user interest function
$\int: U \rightarrow 2^{T}$.
\item[Feasible solutions:] Any set of edges $E \subseteq \{\{u,v\} \mid u,v\in U \land~u\neq v\}$ 
such that the graph $(U,E)$ is topic-connected.
\item[Costs:] Size of $E$.
\item[Goal:] Minimization.
\end{description}
\end{problem}

In this paper we study also some of its special instances. We restrict the number of users
that are interested in a common topic, \ie the size of $U_t$, to a constant.
We also study the instances where each user is interested in a constant number 
of topics.
The definitions necessary for these special instances are summarized in the beginning
of the corresponding section.

We refer here to the famous {\em minimum hitting set problem} (\minHS) and {\em minimum
set cover problem} (\minSC). In \minHS, we are given a system of sets 
${\cal S} = \{S_1, \dots, S_m\}$ on $n$ elements $X = \{x_1, \dots, x_n\}$ (\ie $S_j \subseteq X$).
A feasible solution of this problem is a set $H \subseteq X$, such that $S_j \cap H \neq \emptyset$ 
for all $j$. Our goal is to minimize the size of $H$.
The \minSC is the dual problem to \minHS. In this problem, we are given a system of sets 
${\cal S} = \{S_1, \dots, S_m\}$ on $n$ elements $X = \{x_1, \dots, x_n\}$, a feasible solution
is a set $S \subseteq {\cal S}$ of sets such that for all $i$ there exists $j$ such that
$x_i \in S_j \in S$ and the goal is the minimization of the size of $S$.

There are many modifications and subproblems of the hitting set problem that are intensively studied.
In our paper, we refer to the \dHS problem -- a~restriction of \minHS to instances where
$|S_i| \le d$ for all $i$.

The \minHS is equivalent to the \minSC (\cite{AAP80}), 
thus all the properties of \minSC carry over to \minHS. Following from these properties,
we have \logapx-completeness of \minHS (\cite{CKS01}) and \apx-completeness of \dHS (\cite{PY91}).
There is a well known $d$-approximation algorithm for \dHS (\cite{BE81}), it can be approximated with ratio 
$d - \frac{(d-1)\ln\ln n}{\ln n}$ (\cite{Hal02}), it is \classNP-hard to approximate it within
a~factor ($d-1-\varepsilon$) (\cite{DGKR03}) and \dHS is not approximable within a~factor better 
than $d$, unless the unique games conjecture fails (\cite{KR08}).

We use the standard definitions from complexity theory (for details see~\cite{Hro03}):
\begin{itemize}
\item For \npo problems in the class \ptas, there exists an algorithm that, for arbitrary
$\varepsilon>0$, produces a solution in time polynomial in the input size (but possibly 
exponential in $1/\varepsilon$ that is within a~factor $(1+\varepsilon)$ from optimal.
\item The \npo problems in the class \apx are approximable by some constant-factor approximation 
algorithm in polynomial time.
\item For \npo problems in the class \logapx, there exists a polynomial-time logarithmic-factor
approximation algorithm.
\end{itemize}

Thus $$\ptas \subseteq \apx \subseteq \logapx\mathrm{.}$$

\begin{definition} \label{def:ptasred}
Let $A$ and $B$ be two \npo minimization problems.
Let $I_A$ and $I_B$ be the sets of the instances of $A$ and $B$, respectively.
Let $S_A(x)$ and $S_B(y)$ be the sets of the feasible solutions and let $\cost_A(x)$ and 
$\cost_B(y)$ be polynomially computable measures of the instances $x \in I_A$ and 
$y \in I_B$, respectively. We say that $A$ is {\em \ap-reducible} to $B$, if there exist 
functions $f$ and $g$ and a constant $\alpha>0$ such that:
\begin{enumerate}
\item For any $x \in I_A$ and any $\varepsilon>0$, $f(x,\varepsilon) \in I_B$.
\item For any $x \in I_A$, for any $\varepsilon>0$, and any $y \in S_B(f(x,\varepsilon))$,
$g(x,y,\varepsilon) \in S_A(x)$.
\item The functions $f$ and $g$ are computable in polynomial time with respect to the sizes of instances
$x$ and $y$, for any fixed $\varepsilon$.
\item The time complexity of computing $f$ and $g$ is nonincreasing with $\varepsilon$ for 
all fixed instances of size $|x|$ and $|y|$.
\item For any $x \in I_A$, for any $\varepsilon>0$, and for any $y \in S_B(f(x,\varepsilon))$
\begin{align*}
\frac{\cost_B(y)}{\min \{\cost_B(z)~|~z \in S_B(f(x,\varepsilon))\}} & \leq 1+\varepsilon~implies\\
\frac{\cost_A(g(x,y,\varepsilon))}{\min\{\cost_A(z)~|~z \in S_A(x)\}}& \leq 1+\alpha\cdot\varepsilon\mathrm{.}
\end{align*}
\end{enumerate}
\end{definition}

\section{Results for \minTCO When The Number of Users Interested in a Common 
Topic is a Constant} \label{sec:hardness-topic-bounded}

In this whole section, we denote by a triple $(U,T,\int)$ an instance of \minTCO.
We focus here on the case where the number of users that share a topic $t$, \ie $\max_{t\in T}|U_t|$, is
bounded. 

We present here a lower bound on the approximability, a constant
approximation algorithm and an \apx-completeness proof for these restricted instances of \minTCO.

\subsection{Hardness results} \label{subsec:hardness}

It is easy to see that, if $\max_{t\in T}|U_t|\le2$, then \minTCO can be solved in linear time, 
because two users sharing a topic $t$ should be directly connected by an edge, which is the
unique minimum solution.
 
\begin{theorem}
If $\max_{t\in T}|U_t| \le 2$, then \minTCO can be solved in linear time.
\end{theorem}

We extend the methods from \cite{CMTV07} and design an \ap-reduction from \dHS to \minTCO,
where $\max_{t\in T}|U_t| \le d+1$.

\begin{theorem} \label{theo:AP-red-dHS-minTCO}
For arbitrary $d\ge2$, there exists an \ap-reduction from \dHS to \minTCO, where $\max_{t\in T}|U_t|\le d+1$.
\end{theorem}
\begin{proof}
Let $\ihs=(X,{\cal S})$ be an instance of \dHS and let $\varepsilon > 0$ be arbitrary. We omit the subscript in
the functions $\cost_{d\mathrm{-HS}}$ and $\cost_{\mathrm{Min-TCO}}$ as they are unambiguous.
For the instance $\ihs$, we create an instance $\itco=(U, T, \int)$ of \minTCO with $\max_{t\in T}|U_t| \le d+1$
with $|X|+k$ users, where $k=|X|^2 \cdot \big\lceil\frac{1+\varepsilon}{\varepsilon}\big\rceil$, as 
follows (the function $f$ in the definition of \ap-reduction).
\begin{align*}
U & = X \cup \{p_i~|~p_i \notin X \land 1 \le i \le k\}, \\
T & = \{t_{S_j}^i~|~S_j \in {\cal S} \land 1 \le i \le k\}, \\
\int(x) & = \begin{cases} 
   \{t_{S_j}^i~|~x\in S_j \land S_j \in {\cal S} \land 1 \le i \le k\} & \quad \mathrm{for~} x \in X \\ 
   \{t_{S_j}^i~|~S_j \in {\cal S}\} & \quad \mathrm{for~} x=p_i  
 \end{cases}
\end{align*}
 
Observe that the instance contains $k \cdot |{\cal S}|$ topics and its size is polynomial in the size of $\ihs$.
The users interested in a topic $t^i_{S_j}$ ($S_j \in {\cal S}$) are
exactly the elements that are members of set $S_j$ in \dHS plus a {\em special user $p_i$} ($1 \le i \le k$).
Let $\stco$ be a feasible solution of \minTCO on instance $\itco$.
We partition the solution into levels. Level $i$ is a set $L_i$ of the edges of $\stco$ that are incident
 with the special user $p_i$. In addition, we denote by $L_0$ the set of edges of $\stco$ that
are not incident with any special user. Therefore, $\stco = \bigcup_{i=0}^{k} L_i$
and $L_i \cap L_j = \emptyset$ ($0 \le i < j \le k$).
 
We claim that, for any $L_i$ ($1 \le i \le k$), the set of the non-special users incident with edges of $L_i$
is a feasible solution of the instance $\ihs$ of \dHS. This is true since, if a set $S_j \in {\cal S}$ is not hit,
none of the edges $\{x,p_i\}$ ($x \in S_j$) is in $L_i$.
But then the users interested in topic $t^i_{S_j}$ are not interconnected as user $p_i$ is disconnected.
 
Let $j$ be chosen such that $L_j$ is the smallest of all sets $L_i$, for $1 \le i \le k$.
We construct $\shs$ by picking all the non-special users that are incident to some edge from $L_j$
(the function $g$ in the definition of \ap-reduction).
Denote an~optimal solution of \dHS and \minTCO for $\ihs$ and $\itco$ by $\ohs$ and $\otco$, respectively.
 
If we knew $\ohs$, we would be able to construct a feasible solution of \minTCO on $\itco$ as follows.
First, we pick the edges $\{x,p_i\}$, $x \in \ohs$, for all special users $p_i$, and include them in
the solution.
This way, for any topic $t \in \int(p_i)$, we connect $p_i$ to some element of $X$ that
is interested in $t$, too. To have a feasible solution, we could miss some edges between some
elements of $X$. So, we pick all the edges between elements from $X$. The feasible
solution of \minTCO on $\itco$ that we obtain has roughly cost
$$k \cdot \cost(\ohs) + |X|^2 \ge \cost(\otco)\mathrm{.}$$
 
On the other hand, if we replace all levels $L_i$ ($1\le i \le k$) by level $L_j$ in $\stco$, we
still have a feasible solution of \minTCO on $\itco$, with cost possibly smaller. Thus
$$k \cdot \cost(\shs) \le \cost(\stco)\mathrm{.}$$
 
We use these two inequalities to bound the cost of $\shs$:
$$k \cdot\cost(\shs) \le \frac{\cost(\stco)}{\cost(\otco)} \cdot \left(k \cdot \cost(\ohs) + |X|^2\right)$$
and thus
$$\frac{\cost(\shs)}{\cost(\ohs)} \le \frac{\cost(\stco)}{\cost(\otco)} \cdot 
\left(1 + \frac{|X|^2}{k}\right)\mathrm{.}$$
 
If $\cost(\stco)/\cost(\otco) \le 1+\varepsilon$ and $\alpha:=2$, then we have
\begin{align*}
\frac{\cost(\shs)}{\cost(\ohs)} & \le (1+\varepsilon) \cdot\left(1 +\frac{|X|^2}{k}\right) \le
(1+\varepsilon) \cdot\left(1+\frac{\varepsilon}{1+\varepsilon}\right) = 1+2\varepsilon\mathrm{.}
\end{align*}
 
It is easy to see that the five conditions of Definition~\ref{def:ptasred} are satisfied and thus we
have an \ap-reduction.
\qed
\end{proof}
 
\begin{corollary}
For any $\delta>0$ and polynomial-time $\alpha$-approximation algorithm of \minTCO with
$\max_{t\in T}|U_t| \le d+1$, there exists a polynomial-time $(\alpha+\delta)$-approximation
algorithm of \dHS.
\end{corollary}
\begin{proof}
The approximation algorithm for \dHS would use Theorem~\ref{theo:AP-red-dHS-minTCO} with
$k:=|X|^2 \cdot \lceil\frac{\alpha}{\delta}\rceil$.
\qed
\end{proof}

Our theorem also implies the following negative results on approximability. 
One of them holds if {\em unique games conjecture} is true. This conjecture is discussed,
for example, in \cite{WS11} and was introduced by Khot in \cite{Kho02}.

\begin{corollary}
\minTCO with $\max_{t\in T}|U_t| \le d$ ($d\ge3$) is \classNP-hard to approximate within a factor
of $(d-1-\varepsilon)$, for any $\varepsilon>0$, and, if the unique games conjecture holds, there is no 
polynomial-time \mbox{$(d-\varepsilon)$}-approximation algorithm for it.
\end{corollary}
\begin{proof}
Otherwise, the reduction described in the proof of Theorem~\ref{theo:AP-red-dHS-minTCO} 
would imply an approximation algorithm for \dHS with a ratio better than $d-1$ and $d$ 
respectively. This would directly contradict theorems proven in~\cite{DGKR03} and \cite{KR08}.
\qed
\end{proof}

The following corollary is an improvement of the already known results 
of \cite{CMTV07} where an $O(\log |T|)$-approximation algorithm is presented,
and of \cite{AAR10} where a lower bound of $\Omega(\log(n))$ on the approximability
is shown. We close the gap by designing a reduction that can reduce
{\em any} problem from class \logapx to \minTCO preserving the approximation
ratio up to a constant.

\begin{corollary}
\minTCO is \logapx-complete.
\end{corollary}
\begin{proof}
\minTCO is in the class \logapx since it admits a logarithmic approximation algorithm as 
presented in~\cite{CMTV07}. Our reduction from the proof of 
Theorem~\ref{theo:AP-red-dHS-minTCO} is independent of $d$ and thus an \ap-reduction from 
\logapx-complete \minHS to \minTCO.
\qed
\end{proof}

\subsection{A Constant Approximation Algorithm} \label{subsec:approx}

In this subsection, we present a reduction from \minTCO with $\max_{t\in T}|U_t|\le d$ to \oHS thus 
showing that there exists a constant approximation algorithm for \minTCO with $\max_{t\in T}|U_t|\le d$ 
as \dHS is constantly approximable. Moreover, the constant approximation algorithm classifies this
problem to be a member of the class \apx and thus, since the \apx-hardness was proven
in Subsection~\ref{subsec:hardness}, we conclude that \minTCO with $\max_{t\in T}|U_t|\le d$ is \apx-complete.

Recall that a partition of vertices $V$ in graph $G$ is a tuple $(A,B)$, such that $A\subseteq V$,
$B\subseteq V$, $A\cap B = \emptyset$, and $A\cup B = V$.

\begin{definition}
Let $V=\{v_1, \dots, v_n\}$ be a set of vertices and for every partition $(A_i,B_i)$ of $V$, let
$E_i=\{\{u,v\}\mid u\in A_i \land v\in B_i\}$. Then we call the system ${\cal S} = \{E_1, \dots,\penalty0 E_m\}$
of all sets of edges between vertices of all the partitions of $V$ 
\emph{a characteristic system} of edges on $V$.
In other words, ${\cal S}$ contains all sets of edges that form a maximum bipartite graph on $V$.
\end{definition}

In the following lemma, we show the basic properties of characteristic systems of edges.

\begin{lemma} \label{lemma:charact}
Let ${\cal S} = \{E_1, \dots, E_m\}$ be a characteristic system of edges on the set $V$ of $n$ vertices. Then
\begin{enumerate}
\item \label{enum:iff1} $m=2^{n-1}-1$.
\item \label{enum:iff2} $|E_j| \leq\lfloor n/2 \rfloor \cdot \lceil n/2 \rceil$, for all $j$, $1 \le j \le m$.
\item \label{enum:iff3} Any two sets $E_i$ and $E_j$ differ in at least $n-1$ elements ($1 \le i < j \le m$).
\item \label{enum:iff4} $H \subseteq \{\{u,v\}~|~u,v\in V \land u\neq v\}$ is a hitting set
of $(\{\{u,v\}~|~u,v\in V\land \penalty0 u\neq v\},{\cal S})$ if and only if $(V,H)$ is connected.
\item \label{enum:iff5} The size of ${\cal S}$ is minimal such that part~\ref{enum:iff4} holds.
\end{enumerate}
\end{lemma}
\begin{proof}
Observe that the complementary graph $(V, F_j)$ ($F_j=\{\{u,v\}~|~u,v\in V \land u\neq v\}\setminus E_j$)
contains two complete graphs -- one on the vertices of $A_j$ and other on the vertices of $B_j$,
and it is a maximal graph (in the number of edges) that is not connected.
We use this observation to prove the last two parts of our lemma.

Part~\ref{enum:iff1}:
We count the different partitions $(A_j,B_j)$ of the vertices $V$ as each such partition determines
a different set $E_j$ of edges. 
There are $2^n$ ways how to distribute vertices from $V$ into partitions. We have to subtract 2 possibilities
for the cases where one of $A_j$ or $B_j$ is empty. Each of the other possibilities
is counted twice -- once when the vertices are present in $A_j$ and once when they are present in $B_j$.

Part~\ref{enum:iff2}:
Let the two sets of vertices $A_j$ and $B_j$ of a partition contain $k>0$ and $n-k$ vertices.
Then the size of $E_j$ is $k\cdot(n-k)$. This function reaches its maximum for $k=n/2$ and thus we can conclude
that, for all $j$, $1\le j \le m$, we have 
$|E_j| \leq \lfloor n/2\rfloor \cdot (n-\lfloor n/2 \rfloor) = \lfloor n/2 \rfloor \cdot \lceil n/2 \rceil$.

Part~\ref{enum:iff3}:
Let us consider two different partitions $(A_i,B_i)$ and $(A_j,B_j)$ of the vertices $V$.
The sets $A_i$ and $A_j$ must differ by at least one vertex. W.l.o.g., let the vertex $v \in A_i$
and $v \notin A_j$. Then, due to the transition of the vertex $v$ from $A_i$ to $B_j$, 
there are $|B_i|$ edges that are in $E_i$ but cannot be in $E_j$, and there are $|A_i|-1$ 
edges that are not in $E_i$, but are in $E_j$. Thus, the overall difference in the number 
of elements between the sets $E_i$ and $E_j$ is at least $|A_i|+|B_i|-1 = n-1$.

Part~\ref{enum:iff4}:
First, we prove the if case. Suppose that $H$ is a hitting set, but $(V,H)$ is not connected.
Since ${\cal S}$ contains complements of all maximal sets of edges that induce a disconnected graph, there exists $j$ ($1\le j\le m$) such that $H \subseteq F_j$.
But then, since $E_j$ is complementary to $F_j$, it follows that $E_j \cap H = \emptyset$. 
Thus, $H$ cannot be a hitting set as $E_j$ is not hit.

For the only-if case, suppose that $(V,H)$ is connected, but $H$ is not a hitting set of
$\left(\{\{u,v\}~|~u,v\in V\land u\neq v\}, {\cal S}\right)$.
Then there exists $j$ such that $E_j$ is not hit by $H$ and thus $H\subseteq F_j$. Yet in such a case, by our
assumption, $(V, F_j)$ is not connected and thus $(V, H)$ cannot be connected as well.

Part~\ref{enum:iff5}: Let ${\cal S'}={\cal S}\setminus E_j$,
$\left(\{\{u,v\}~|~u,v\in V\land u\neq v\}, {\cal S'}\right)$ be an instance of \minHS. Then we claim that $F_j$
is a hitting set of $(\{\{u,v\}~|~u, v\in V \land u\neq v\},\penalty0 S')$.
First, observe that $F_j\neq\emptyset$ since $E_j$ cannot contain all the edges.
Moreover, there exists $e \in E_i$ ($E_i\in {\cal S'}$) such that $e\notin E_j$.
Then $e \in \{\{u,v\}~|~u,v\in V\land u\neq v\}\setminus E_j = F_j$ and thus $F_j$ is a hitting set.
However, by the definition of $F_j$, $(V,F_j)$ cannot be connected and thus, the if case of
part~\ref{enum:iff4} does not hold.
\qed
\end{proof}

Now we are ready to present a simple one-to-one reduction of \minTCO with $\max_{t\in T}|U_t|\le d$ to \oHS. 
The core concept is to construct a system of sets that has to be hit in \oHS as a union over 
all the topics of the characteristic systems of edges on the vertices interested in the topic.

\begin{theorem} \label{theo:dtco2ddhs}
There exists a one-to-one reduction of instances of \minTCO with $\max_{t\in T}|U_t|\le d$ to instance of \oHS.
\end{theorem}
\begin{proof}
Let $\itco=(U,T,\int)$ be an instance of \minTCO with $\max_{t\in T}|U_t|\le d$. For each topic $t\in T$ we define ${\cal S}_t$ to be the
characteristic system of edges on vertices in
$U_t$. Note that Lemma~\ref{lemma:charact} holds for each ${\cal S}_t$
with $n:=d$. We construct an \oHS instance $\ihs=(X,{\cal S})$ as follows:
\begin{align*}
X & = \{\{u,v\}~|~u,v\in U\land u\neq v\}\\
{\cal S} & = \bigcup_{t\in T} {\cal S}_t\mathrm{.}
\end{align*}
The system contains $|U| \choose 2$ elements and at most $|T|\cdot \left(2^{d-1}-1\right)$ sets in $\cal S$
and thus has a size polynomial in $|\itco|$. Obviously, the construction of $\ihs$ takes time polynomial
in $|\itco|$, too. We now show that a feasible
solution of $\itco$ corresponds to a feasible solution of $\ihs$ and vice versa.

First, consider a feasible solution $\shs$ of $\ihs$ and a topic $t\in T$. Due to our construction,
the system ${\cal S}$ contains the characteristic system ${\cal S}_t$ on vertices $U_t$. Therefore,
by Lemma~\ref{lemma:charact} part~\ref{enum:iff4} and the fact that $\shs$ is a hitting set,
we know that the graph induced by the edges in $\shs$ on vertices $U_t$ is connected.

Now, consider a feasible solution $\stco$ of $\itco$. By the following argument, we can easily
see that $\stco$ hits all the sets in $\cal S$. Let $P \in {\cal S}$ be a set that is not hit by
$\stco$. Then there exists $t$ such that $P \in {\cal S}_t$ and thus
a set of the characteristic system was not hit and $\stco$ is not a hitting set of ${\cal S}_t$.
Yet in such a case, considering Lemma~\ref{lemma:charact} part~\ref{enum:iff4}, the subgraph induced on
vertices $U_t$ by edges from $\stco$ cannot be connected and that is in contradiction with the definition of \minTCO
with $\max_{t\in T}|U_t|\le d$.
\qed
\end{proof}

\begin{theorem} \label{theo:approxdtco}
There exists a polynomial-time $\left(\lfloor d/2 \rfloor \cdot \lceil d/2\rceil\right)$-approxima\-tion algorithm 
for \minTCO with $\max_{t\in T}|U_t|\le d$.
\end{theorem}
\begin{proof}
We employ the reduction from Theorem~\ref{theo:dtco2ddhs} together with the well-known 
$d$-approxi\-mation algorithm for \dHS.
Since the size of each set in ${\cal S}$ is at most $\lfloor d/2\rfloor \cdot \lceil d/2 \rceil$ (Lemma~\ref{lemma:charact}
part~\ref{enum:iff2}), by application of this approximation algorithm on \oHS instance $(X,{\cal S})$ we obtain
a $\lfloor d/2\rfloor \cdot \lceil d/2 \rceil$ approximate solution of our \minTCO instance with $\max_{t\in T}|U_t|\le d$.

Note that our reduction is tight in the size of ${\cal S}$ as it is minimal (Lemma~\ref{lemma:charact} part~\ref{enum:iff5}), thus
to achieve an improvement in the approximation algorithm, a different method has to be developed.
\qed
\end{proof}

\begin{corollary}
\minTCO with $\max_{t\in T}|U_t|\le 3$ inherits the approximation hardness of \minVC.
\end{corollary}

\begin{corollary}
\minTCO with $\max_{t\in T}|U_t| \le d$ is \apx-complete, for arbitrary $d\ge3$.
\end{corollary}
\begin{proof}
The \apx-hardness follows from the \apx-hardness of \dHS (\cite{PY91}). Due to our reduction
the problem belongs to the class \apx.
\qed
\end{proof}

\subsection{\minTCO and Parametrized Complexity Theory} \label{sec:param}

In this subsection, we shortly summarize the consequences of our reduction from 
Theorem~\ref{theo:dtco2ddhs} for the field of parametrized complexity, namely we 
present an exact algorithm and a kernelization for \minTCO with $\max_{t\in T}|U_t|$ 
bounded by a constant.

In the research area of exact algorithm design, one searches for an exact solution in exponential time. 
The main goal is to make the base of the exponentiation as small as possible.

A kernelization is a process in which an instance is reduced to a smaller instance 
in polynomial time. Then, instead of solving the original
instance, it is sufficient to solve the problem on the smaller one and then, in polynomial
time, transform its solution back to the initial instance.

\begin{problem} \label{def:pardtco}
\dTCOk is the following parametrized problem:
\begin{description}
\item[Input:] Instance of \minTCO with $\max_{t\in T}|U_t| \le d$ and a parameter $k$.
\item[Goal:] A feasible solution of the \minTCO instance of size at most $k$.
\end{description}
\end{problem}

\medskip

\begin{problem} \label{def:pardhs}
\dHSk is the following parametrized problem:
\begin{description}
\item[Input:] Instance of \dHS and a parameter $k$.
\item[Goal:] A feasible solution of the \dHS instance of size at most $k$.
\end{description}
\end{problem}

We first transform the given instance of \minTCO with $\max_{t\in T}|U_t| \le d$ into an instance of 
\oHS as in Theorem~\ref{theo:dtco2ddhs} and then we apply the kernelization from \cite{Khz07} to 
obtain a kernel of \dTCOk or the exact algorithm from \cite{NR07} to obtain the
first nontrivial exact algorithm for solving \dTCOk.

\begin{theorem}
\dTCOk has a kernel of size $(2c-1)\cdot k^{c-1}+k$ with $c=\lfloor d/2 \rfloor\cdot\lceil d/2\rceil$.
\end{theorem}

\begin{theorem}
\dTCOk on $n$ vertices can be solved in time $O(c^k+n^2)$ with $c = \lfloor d/2 \rfloor\cdot\lceil d/2 \rceil - 1+O(d^{-2})$.
\end{theorem}

\section{Hardness of \minTCO When the Number of Connections of a User is 
Constant}\label{sec:hardness-user-bounded}

It is natural to consider \minTCO with bounded number of connections per user,
\ie to bound $\max_{u \in U} |\int(u)|$, since the number of topics in which one user 
is interested in is usually not too large. 
We show that, sadly, \minTCO is \apx-hard even if $\max_{u \in U} |\int(u)| \le 6$. 
To show this, we design a reduction from {\em minimum vertex cover} (\minVC) to \minTCO.
The minimum vertex cover problem is just a different name for \dHS with $d=2$.
For a better presentation, in this section, we refer to \minVC instead of \tHS.

Given is a graph $G=(V',E')$ and a positive integer $k$ as an instance of 
\minVC, where the goal is to decide whether the given graph has a solution of size at most $k$.
We construct an instance of \minTCO as follows. Let $V=V^{(1)} \cup V^{(2)}$ be the 
set of users, where $V^{(1)}=\{v^{(1)} \mid v\in V'\}$ and 
$V^{(2)}=\{v^{(2)} \mid v\in V'\}$. 
For each edge $e\in E'$, we prepare three topics, $t_e^{(0)}, t_e^{(1)}$ and $t_e^{(2)}$. 
The set of topics is the union of all these topics, \ie 
$T=\bigcup_{e\in E'}\{t_e^{(0)}, t_e^{(1)}, t_e^{(2)}\}$. 
The user interest function $\int$ is defined as 
\begin{align*}
\int(u^{(1)}) & = \bigcup_{e\in E'[N[u]]} \{t_e^{(0)},t_e^{(1)}\} \\
\int(u^{(2)}) & = \bigcup_{e\in E'[N[u]]} \{t_e^{(0)},t_e^{(2)}\}\mathrm{.}
\end{align*}
The following lemma shows the relation between the solutions of the two problems. 

\begin{lemma} \label{lemma:transf-sol-TCO-VC}
The instance $(V,T,\int)$ of \minTCO defined as above has an optimal solution of 
cost $k+2|E'|$ if and only if the instance $(V',E')$ of \minVC has an optimal 
solution of cost $k$. 

\noindent Moreover, any feasible solution $H$ of $(V,T,\int)$ can be transformed 
into a feasible solution of $(V',E')$ of cost at most $|H|-2|E'|$.
\end{lemma}
\begin{proof}
It is obvious that any feasible solution $H$ of the instance of \minTCO contains
the edge $\{u^{(i)}, v^{(i)}\}$ ($i\in\{1,2\}$), for every edge $e=\{u,v\}$,
because only $u^{(1)}$ and $v^{(1)}$ (resp., $u^{(2)}$ and $v^{(2)}$) are
interested in topic $t_e^{(1)}$ (resp., $t_e^{(2)}$).

Since each feasible solution $H$ of $(V,T,\int)$ must contain the edges
$\{u^{(1)},\penalty-1000 v^{(1)}\}$ and $\{u^{(2)},v^{(2)}\}$, for every edge $e=\{u,v\}\in E'$,
it is sufficient to consider only the topics $t_e^{(0)}$.
The number of edges in $H$ connecting a vertex from $V^{(1)}$ with a vertex from
$V^{(2)}$ is at most $|H|-2|E'|$.

For an edge $e=\{u,v\}$, the vertices that are interested in $t_e^{(0)}$ are
$u^{(1)}, v^{(2)}$, $v^{(1)}$ and $u^{(2)}$. Since these four vertices have
to be connected, $H$ contains at least one edge of $\{u^{(1)},u^{(2)}\}$,
$\{v^{(1)},v^{(2)}\}$, $\{u^{(1)},v^{(2)}\}$ and $\{v^{(1)},u^{(2)}\}$.

The optimal solution of $(V,T,\int)$ contains at most two of these four
edges, namely the edges $\{u^{(1)},u^{(2)}\}$ and $\{v^{(1)},v^{(2)}\}$.
Observe that, for {\em each} edge $f$ that is incident with vertex $u$ in $G$,
the edge $\{u^{(1)},u^{(2)}\}$ connects the solution to be $t_f^{(0)}$-connected.
The only topic that the other two edges connect is $t_{\{u,v\}}^{(0)}$ and
thus they can be replaced by $\{u^{(1)},u^{(2)}\}$ or $\{v^{(1)},v^{(2)}\}$.

In any non-optimal solution, more than two of the four edges may be present and
the replacement of edges $\{u^{(1)},v^{(2)}\}$ and $\{v^{(1)},u^{(2)}\}$ by
$\{u^{(1)},u^{(2)}\}$ and $\{v^{(1)},v^{(2)}\}$, respectively, may lead
to a decrease of the cost of the solution.

We assume that the solution of $(V,T,\int)$ has been transformed
so that it does not contain cross edges between $u^{(i)}$ and $v^{(3-i)}$
($i\in\{1,2\}$).
The vertices that correspond to the edges between the two layers
$V^{(1)}$ and $V^{(2)}$ form a feasible solution of \minVC.
As discussed above, its size is at most $|H|-2|E'|$ for a feasible
solution $H$ and exactly $|H|-2|E'|$ for an optimal solution $H$.
This proves one implication of the first claim and the second claim.

We show that, if \minVC has an optimal solution of size $k$,
then the instance of \minTCO has an optimal solution of size $k+2|E'|$.
From an optimal solution $W\subseteq V'$ of \minVC, we construct the optimal
solution of \minTCO as
$H=\{\{u^{(1)},v^{(1)}\}, \{u^{(2)},v^{(2)}\} \mid \{u,v\}\in E'\} \cup 
\{\{u^{(1)},u^{(2)}\}\mid u \in W\}$. Clearly, the size of $H$ is exactly
$k+2|E'|$. As $W$ is the smallest set of vertices that covers all the
edges of $E'$, its corresponding edges of \minTCO produce the minimal
set of edges that connect every topic with superscript 0. Thus,
$H$ satisfies the connectivity requirement for every topic $t\in T$ and is optimal.
\qed
\end{proof}

We use the \minVC on degree-bounded graphs, which is \apx-hard, to show
lower bounds for our restricted \minTCO.
By the above reduction and the lemma, we prove the following theorem. 

\begin{theorem}
\minTCO with $\max_{v\in U} |\int(v)|\le 6$ cannot be approximated within
a factor of $694/693$ in polynomial time, unless $\classP = \classNP$, even if 
$|\int(v) \cap \int(u)|\le 3$ holds for every pair of different users $u, v\in U$.
\end{theorem}
\begin{proof}
We prove the statement by contradiction. Suppose that there exists an approximation algorithm
$A$ for \minTCO with the above stated restrictions that has the ratio $(1+\delta)$.

Let $G=(V',E')$ be an instance of \minVC and let $G$ be cubic and regular (\ie
each vertex is incident with exactly three edges).
We construct an instance $\itco$ of \minTCO as stated above and we apply our  
algorithm $A$ to it to obtain a feasible solution $\stco$. From such a solution, by 
Lemma~\ref{lemma:transf-sol-TCO-VC}, we create a~feasible solution of the original 
\minVC instance $\svc$. 
We denote by $\otco$ and $\ovc$ the optimal solutions of $\itco$ and $G$, respectively.

Let $d$ be a constant such that $d\cdot\cost(\ovc) = 3|V'|$. Since $G$ is cubic and regular,
$\cost(\ovc) \ge |E'|/3 = |V'|/2$ and thus $d \le 6$.

Observe that, due to Lemma~\ref{lemma:transf-sol-TCO-VC}, 
$\cost(\otco) = \cost(\ovc) + 2|E'| = \cost(\ovc)+d\cdot\cost(\ovc)$
and $\cost(\stco) \ge \cost(\svc) + 2|E'| = \penalty-10000\cost(\svc) + d\cdot\cost(\ovc)$.
These two estimations give us the following bound

$$\frac{\cost(\svc)+d\cdot\cost(\ovc)}{\cost(\ovc)+d\cdot\cost(\ovc)} \le \frac{\cost(\stco)}{\cost(\otco)} 
\le 1+\delta\mathrm{.}$$

The above inequality allows us to bound the ratio of our \minVC solution $\svc$ and the optimal
solution $\ovc$:
$$\frac{\cost(\svc)}{\cost(\ovc)} \le (1+\delta)\cdot(d+1)-d = 1 + \delta(d+1) \le 1+7\delta\mathrm{.}$$

For $\delta:=\frac{1}{693}$, we obtain a $\frac{100}{99}$-approximation algorithm for \minVC on 
3-regular graphs which is directly in contradiction with a theorem proven in~{\cite{CC03}}.\\
\qed
\end{proof}

\begin{corollary}
\minTCO with $\max_{v\in U} |\int(v)| \le 6$ is \apx-hard.
\end{corollary}

\begin{corollary}
\minTCO with $|\int(v) \cap \int(u)|\le3$, for all users $u,v \in U$, is \apx-hard.
\end{corollary}

This result is almost tight, the case when $|\int(v) \cap \int(u)| \le 2$ is
still open. The following theorem shows that \minTCO with 
$|\int(v) \cap \int(u)|\le 1$, for every pair of distinct users $u, v\in U$, 
can be solved in linear time. 

\begin{theorem}
\minTCO can be solved in linear time, if $|\int(v) \cap \int(u)|\le 1$ holds
for every pair of users $u, v\in U$, $u\neq v$.
\end{theorem}
\begin{proof}
We execute the following simple algorithm. First set the solution
$E:=\emptyset$. Then sequentially, for each topic $t$, choose its representative
$v^* \in U$ ($t \in \int(v^*)$) and add edges 
$\{\{v^*,u\}\mid u\in U_t\setminus \{v^*\}\}$ to the solution $E$.
We show that, if $|\int(v) \cap \int(u)|\le 1$, for all distinct $u,v\in U$, 
then the solution $E$ is optimal.

Observe that, in our case, any edge in any feasible solution is present 
because of a unique topic. We cannot find an edge $e=\{u,v\}$ of the
solution that belongs to the subgraphs for two different topics.
(Otherwise $|\int(v) \cap \int(u)|>1$ and our assumption would be wrong 
for the two endpoints of the edge $e$.)
Thus, any solution consisting of spanning trees for every topic is feasible
and optimal. Note that its size is $|T|\cdot(|U|-1)$.
\qed
\end{proof} 

\begin{corollary}
\minTCO with $\max_{u\in U}|\int(u)| \le 2$ can be solved in linear time. 
\end{corollary}

\section{A Polynomial-Time Algorithm for \minTCO with Bounded Number of Topics}\label{sec:bounded-number-of-topics}

In this section, we present a simple brute-force algorithm that achieves
a polynomial running time when the number of topics is bounded by
$|T| \le \ll|U|-\frac{3}{2}\lll|U|$.

\begin{theorem}
The optimal solution of \minTCO can be computed in polynomial time if 
$|T|\leq (1 + \varepsilon(|U|))^{-1} \cdot\log\log |U|$, for a function
$$\varepsilon(n)\ge\frac{3/2\lll n}{\ll n-3/2\lll n}\mathrm{.}$$
\end{theorem}
\begin{proof}
Let $(U,T,\int)$ be an instance of \minTCO such that $|T|\leq (1 + \varepsilon(|U|))^{-1}\cdot\log\log |U|$.
Moreover, $|T|>2$, otherwise the problem is solvable in polynomial time.
We shorten the notation by setting $t=|T|$ and $n=|U|$.

First observe that, if $u,v\in U$ and $\int(u)\subseteq\int(v)$, instead of solving instance
$(U,T,\int)$, we can solve \minTCO on instance $(U\setminus \{u\},T,\int)$ and add to such
solution the direct edge $\{u,v\}$. 
Note that $u$ has to be incident with at least one edge in any solution. Thus, the addition
of the edge $\{u,v\}$ cannot increase the cost. 
Moreover, any other user that would be connected to $u$ in some solution can be also
connected to $v$. Thus, we can remove $u$, solve the smaller instance and then add $u$ by
a single edge. Such a solution is feasible and its size is unchanged.
We say that vertex $u$ is {\em dominated} by the vertex $v$ if $\int(u)\subseteq \int(v)$.

Therefore, before applying our simple algorithm, we remove from the instance
all the users that are dominated by some other user. We denote the set of remaining
users (\ie those with incomparable sets of interesting topics) by $M$.
The largest system of incomparable sets on $n$ elements is called a Sperner system and
it is a well known fact that its size is at most ${n \choose \lfloor n/2 \rfloor}$.
Since every user in $M$ must have different set of interesting topics and these sets are
all incomparable, we have $$m=|M| \le {t \choose \lfloor t/2 \rfloor} \le 
\frac{2^t}{\sqrt{t}}\mathrm{.}$$ (To verify the bound, consider $t$ to be odd or even
and use ${2n \choose n} \le \frac{4^n}{\sqrt{3n+1}}$, $n \ge 1$.)

Our simple algorithm exhaustively searches over all the possible solutions on instance
$(M,T,\int)$ and then reconnects each of the removed users $U\setminus M$ by a single edge.
The transformation to set $M$ and the connection of the removed users is clearly
polynomial. Thus, we only need to show that our exhaustive search is polynomial.

Observe that the size of the optimal solution is at most $t(m-1)$, as
merged spanning trees, for all the topics, form a feasible solution.
Our algorithm exhaustively searches over all possible solutions,
\ie it tries every possible set of $i$ edges for $1\le i \le t(m-1)$
and verifies the topic-connectivity requirements for such sets of edges.
The verification of each set can be done in polynomial time. The number
of sets it checks can be bounded as follows:
\begin{eqnarray*}
\sum^{t(m-1)}_{i=1} {{m \choose 2} \choose i} & \leq & \sum^{tm}_{i=1} {m^2 \choose i} \\
~ &\leq& tm\cdot {m^2 \choose tm} \\
~ &\leq& tm\cdot m^{tm} \\
~ &\leq& m^{tm} \cdot O(\log^2 n) 
\end{eqnarray*}
(Note that $tm \le m^2/2$ and thus the binomial coefficient is maximal in $tm$. Otherwise
the number of all possible choices of edges into a solution is polynomial in $n$.)

To check a polynomial number of sets, it is sufficient to bound the factor $m^{tm}$ by a polynomial, 
\ie by at most $n^c$ for some $c>0$. (In all our calculations, $\log$ stands for the binary logarithm,
however any other logarithm can be used as the change will effect the exponent by a constant.)
We consider two cases:
\begin{enumerate}
\renewcommand{\theenumi}{\Alph{enumi}}
\renewcommand{\labelenumi}{\theenumi:}
\item First assume that $t \le \frac{\ll n}{1+2\varepsilon(n)}$, then $m\le2^t\le(\log n)^{(1+2\varepsilon(n))^{-1}}$.

We use the upper bounds on $t$ and $m$ to estimate the number of sets our exhaustive search has to check:
$$m^{tm}\le (\log n)^{(1+2\varepsilon(n))^{-2}\cdot\ll n \cdot (\log n)^{(1+2\varepsilon(n))^{-1}}} \le n^c\mathrm{.}$$
Then we take the logarithm of the inequality, leading to
\begin{align*}
(1+2\varepsilon(n))^{-2}\cdot\ll^2 n & \le c\cdot(\log n)^{\frac{2\varepsilon(n)}{1+2\varepsilon(n)}}\mathrm{.}\\
\intertext{After another logarithm operation, we obtain the following inequality:}
-2\log(1+2\varepsilon(n)) + 2\lll n & \le \frac{2\varepsilon(n)}{1+2\varepsilon(n)}\cdot\ll n +\log c\mathrm{.}
\end{align*}
We prove inequality~(\ref{egn:bound-on-epsilon}) instead. In the end, we will see that the function
$\varepsilon(n)$ is positive, except for the first few values. Thus, for large inputs,
$2\log(1+2\varepsilon(n))$ is positive and thus the above inequalities will hold, too.

\begin{eqnarray} \label{egn:bound-on-epsilon}
2 \lll n & \le & \frac{2\varepsilon(n)}{1+2\varepsilon(n)}\cdot\ll n
\end{eqnarray}
We are now able to estimate the function $\varepsilon(n)$:
\begin{eqnarray} \label{eqn:final-epsilon}
\varepsilon(n) \ge \frac{\lll n}{\ll n-2\lll n}\mathrm{.}
\end{eqnarray}

Due to the following case, we use $\varepsilon(n) \ge \frac{3/2 \lll n}{\ll n-3/2\lll n}$ that also 
satisfies~(\ref{eqn:final-epsilon}) and is positive for $n\ge 16$.
\item To conclude the proof, assume that $\frac{\ll n}{1+2\varepsilon(n)} < t \le \frac{\ll n}{1+\varepsilon(n)}$.
Since we have both an upper and a lower bound on $t$, we can refine the estimation of $m$:
$$m \le \frac{2^t}{\sqrt{t}} \le (\log n)^{(1+\varepsilon(n))^{-1}}\cdot (1+2\varepsilon(n))^{1/2}\cdot
(\ll n)^{-1/2}\mathrm{.}$$
We show that $m^{tm}$ is polynomial in $n$ similarly as in the previous case:
$$ 
(\log n)^{
  (1+\varepsilon(n))^{-2} \cdot 
  (\ll n)^{1/2} \cdot 
  (\log n)^{
    (1+\varepsilon(n))^{-1}
  }\cdot 
  (1+2\varepsilon(n))^{1/2}
} 
\le m^{tm} \le n^c\mathrm{.}$$

Then we take the logarithm of the inequality, leading to
$$(1+\varepsilon(n))^{-2} \cdot (\ll n)^{3/2} \cdot (1+2\varepsilon(n))^{1/2}
\le 
c \cdot (\log n)^{\frac{\varepsilon(n)}{1+\varepsilon(n)}}\mathrm{.}
$$

Assume that $(1+2\varepsilon(n))^{1/2} \le 1+\varepsilon(n)$, except for the first few values,
then it is sufficient to prove a simpler inequality:

\begin{align*}
(1+\varepsilon(n))^{-1} \cdot (\ll n)^{3/2} &\le c \cdot (\log n)^{\frac{\varepsilon(n)}{1+\varepsilon(n)}}\mathrm{.}\\
\intertext{After another logarithm operation, we obtain the following inequality:}
-\log(1+\varepsilon(n)) + 3/2\cdot \lll n   & \le \frac{\varepsilon(n)}{1+\varepsilon(n)} \cdot\ll n + \log c
\mathrm{.}
\end{align*}

Again, assuming that $\log(1+\varepsilon(n)) > 0$ if $n$ tends to infinity, to prove the above inequality, 
it is sufficient to show that

\begin{eqnarray*} 
3/2 \lll n & \le & \frac{\varepsilon(n)}{1+\varepsilon(n)}\cdot\ll n\mathrm{.}
\end{eqnarray*}

Thus we are able to bound the function $\varepsilon(n)$ as

\begin{eqnarray*} 
\varepsilon(n) \ge \frac{3/2\lll n}{\ll n-3/2\lll n}\mathrm{.}
\end{eqnarray*}

Observe that $\varepsilon(n)>0$ for $n\ge 16$, thus both assumptions that we made hold for $|U|\ge 16$
which concludes the proof.
\end{enumerate}
\qed
\end{proof}

\section{Conclusion} \label{sec:conclusion}

In this paper, we have closed the gap in the approximation hardness of \minTCO by showing 
its \logapx-completeness. We studied a subproblem of \minTCO  
where the number of users interested in a common topic is bounded by a~constant $d$. 
We showed that, if $d\le2$, the restricted \minTCO is in \classP and, if $d\ge3$, it is 
\apx-complete. The latter result, together with the constant approximation algorithm we presented,
allows us to prove lower bounds on approximability of these
special instances that match any lower bound known for any problem from the class \apx.
Furthermore, we studied instances of \minTCO where the number of topics in which a single user
is interested in is bounded by a~constant $d$. We presented a reduction that shows
that such instances are \apx-hard for $d=6$. In this reduction, any two users have
at most three common topics, thus the reduction shows also that \minTCO restricted in this way 
is \apx-hard. We also investigated \minTCO with a bounded number
of topics. Here we presented a polynomial-time algorithm for 
$|T| \le (1 + \varepsilon(|U|))^{-1}\cdot\log\log |U|$ and a function
$\varepsilon(n)\ge\frac{3/2\lll n}{\ll n-3/2\lll n}$. The case where $t=\omega(\log\log n)$ and
$t=o(n)$ remains to be a challenging open problem.

\bibliographystyle{elsarticle-num}

\end{document}